\documentclass[aps,pra,nofootinbib,showpacs,amssymb,superscriptaddress,onecolumn,notitlepage]{revtex4-1}

\usepackage{amsmath,amsfonts,amsthm,amssymb}
\usepackage{amsmath}
\usepackage{verbatim}
\usepackage{bbm}
\usepackage{bbm}
\usepackage{bm}
\usepackage{amsbsy}
\usepackage{graphicx} 
\usepackage{epsfig}
\usepackage{epstopdf}
\usepackage{color}
\usepackage[colorlinks]{hyperref}
\usepackage[figure,table]{hypcap}
\usepackage{enumerate}
\usepackage[resetlabels]{multibib}
\usepackage[normalem]{ulem} 

\newtheorem{theorem}{Theorem}
\newtheorem{lemma}{Lemma}

\newcommand{\be}{\begin{equation}}
\newcommand{\ee}{\end{equation}}
\newcommand{\ba}{\begin{align}}
\newcommand{\ea}{\end{align}}
\newcommand{\tr}{\mathrm{Tr}}

\newcommand{\inp}[2]{\ensuremath{\langle #1|#2 \rangle} }

\newcommand{\proj}[1]{\ensuremath{| #1\rangle\!\langle #1 |} }

\newcommand{\me}{\mathrm{e}}

\hypersetup{
	bookmarksnumbered,
	pdfstartview={FitH},
	citecolor={darkgreen},
	linkcolor={darkred},
	urlcolor={darkblue},
	pdfpagemode={UseOutlines}}
\definecolor{darkgreen}{RGB}{50,190,50}
\definecolor{darkblue}{RGB}{0,0,190}
\definecolor{darkred}{RGB}{238,0,0}

\newcommand{\Tr}{\mathrm{Tr}}

\newcommand{\bra}[1]{\ensuremath{\left\langle#1\right\rvert}}
\newcommand{\ket}[1]{\ensuremath{\left\lvert#1\right\rangle}}

\begin{document}

\title{Work and Reversibility in Quantum Thermodynamics}
\author{\'{A}lvaro M. Alhambra}
\affiliation{University College London, Department of Physics \& Astronomy, London WC1E 6BT, United Kingdom}
\author{Stephanie Wehner}
\affiliation{QuTech, Delft University of Technology, Lorentzweg 1, 2611 CJ Delft, Netherlands}
\affiliation{Centre for Quantum Technologies, National University of Singapore, 117543 Singapore}
\author{Mark M. Wilde}
\affiliation{Hearne Institute for Theoretical Physics,
Department of Physics and Astronomy,
Center for Computation and Technology,
Louisiana State University,
Baton Rouge, Louisiana 70803, USA}
\author{Mischa P. Woods}
\affiliation{University College London, Department of Physics \& Astronomy, London WC1E 6BT, United Kingdom}
\affiliation{QuTech, Delft University of Technology, Lorentzweg 1, 2611 CJ Delft, Netherlands}
\affiliation{Centre for Quantum Technologies, National University of Singapore, 117543 Singapore}

\begin{abstract}
It is a central question in quantum thermodynamics to determine how irreversible is a process that transforms an initial state $\rho$ to a final state $\sigma$, and whether such irreversibility can be thought of as a useful resource. For example, we might ask how much work can be obtained by thermalizing $\rho$ to a thermal state $\sigma$ at temperature $T$ of an ambient heat bath. Here, we show that, for different sets of resource-theoretic thermodynamic operations, the amount of entropy produced along a transition is characterized by how reversible the process is. More specifically, this entropy production depends on how well we can return the state $\sigma$ to its original form $\rho$ without investing any work. At the same time, the entropy production can be linked to the work that can be extracted along a given transition, and we explore the consequences that this fact has for our results. We also exhibit an explicit reversal operation in terms of the Petz recovery channel coming from quantum information theory. Our result establishes a quantitative link between the reversibility of thermodynamical processes and the corresponding work gain. 
\end{abstract}
\maketitle

\section{Introduction}
Quantum thermodynamics is experiencing a renaissance in which ideas from quantum information theory enable us to understand thermodynamics for even the 
smallest quantum systems. Our inability to apply statistical methods to a small number of particles and the presence of quantum coherences make this a challenging undertaking. 
Yet, we are now indeed able to construct very small quantum devices allowing us to probe such regimes experimentally \cite{koski2014,Ignacio2015a,Ignacio2015b}. Theoretical
results studying the efficiency of small thermal machines~\cite{woods2015,tajima2014r,skrzypczyk2011,entenchance_cooling,Brandner2015,mitchelson2015},
catalysis~\cite{Ng2014,lostaglio2014e,Aberg2014}, work extraction~\cite{HO-limitations,Aberg2013,dahlsten2011i,definingworkEisert,gemmer2015s,scully2003e,Binder2015,salek2015}, and the second laws of quantum thermodynamics~\cite{2ndlaws,Lostaglio2015} have furthermore led to the satisfying conclusion that the usual laws of thermodynamics as we know them can be derived from the laws of quantum mechanics in an appropriate limit. 

Here we are concerned with the fundamental problem of how irreversible is the transformation of a state $\rho_S$ to a state $\sigma_S$ of some system~$S$ in the presence of a thermal bath, and how that irreversibility is related to the work that can be extracted through the same transformation.
In this regard, the second laws~\cite{2ndlaws} provide general constraints on these transitions
, which are necessary and sufficient if $\rho_S$ is diagonal in the energy eigenbasis
of the system.  
Special instances of this problem have drawn particular attention, such as \emph{gaining} the maximum amount of work from $\rho_S$ by thermalizing it to the temperature of the surrounding
bath~\cite{HO-limitations}, extracting work from correlations among different subsystems when $\rho_S$ is a multipartite state (see, e.g.~\cite{extworkcorr}), as well as the case when $\sigma_S$
results from a measurement on~$\rho_S$~\cite{philipp2015,Takahiro2008,Jacob2009}. When thinking about \emph{investing} work, one of the most
well studied instances is Landauer's erasure~\cite{landauer}, which is concerned with the amount of energy necessary to take an arbitrary state $\rho_S$ to a pure
state $\sigma_S$.

We adopt the resource theory approach of~\cite{jarzing2000,fernando2013,HO-limitations}, which has the appealing feature of explictly accounting for all energy flows. We will focus on the quantitative features of the irreversibility of thermodynamical processes that take an initial state $\rho_S$ to a final state $\sigma_S$. 
In particular, we here show that a key quantity, namely the decrease of free energy or \emph{entropy production} is related to how well a particular thermodynamical process can be reversed. In turn, this quantity is directly related to how much work can be extracted in the transition $\rho_S \rightarrow \sigma_S$.

\section{Preliminaries}\label{sec:intro}
Let us now describe a prominent class of processes that we will be dealing with, known in the resource theory approach as \emph{thermal operations}~\cite{2ndlaws}. Given a particular fixed temperature $T$, we may access a bath described by a Hamiltonian $H_B$ and thermal state $\hat{\tau}_B = \exp(-\beta H_B)/Z_B$,
where $\beta = 1/(kT)$ is the inverse temperature\footnote{Here, $k$ is the Boltzmann constant.} and $Z_B$ is the partition function. 
Let $H_S$ be the Hamiltonian associated with the system $S$
and let $U$ denote a unitary that acts on the system $S$, a battery system $W$, and the bath $B$.  
The only unitary transformations $U$ that are allowed are those that conserve total energy. That is, the allowed unitaries are such that $[U,H]=0$,
where $H = H_S + H_W  + H_B$ is the total Hamiltonian. The transformation $\mathcal{T}$ performing the mapping $\mathcal{T}(\rho_S \otimes \proj{0}_W) = \sigma_S \otimes \proj{1}_W$ then takes the following form
\begin{equation}\label{eq:TOaction}
\mathcal{T}(\eta_{SW}) = \tr_{B}[U(\eta_{SW} \otimes \hat{\tau}_B)U^\dagger]
\end{equation}
for some input state $\eta_{SW}$ of the system and the battery.
Other classes of thermodynamic operations are discussed in Sections \ref{sec:ext} and \ref{sec:Gibbs preserving map}. Given that $U$ conserves total energy, it is clear that this framework accounts for all energy flows, making it particularly appealing for studying quantum thermodynamics.

We shall focus on the following quantity, known as entropy production
\begin{equation}\label{eq:entprod}
F(\rho_S)-F(\sigma_S),
\end{equation}
where $F(\omega_{S}) = \tr[H_{S} \omega_{S}] - kT \,S(\omega_{S})$ is the Helmholtz free energy and the von Neumann entropy is defined as $S(\omega_S)=-\tr[\omega_S \log \omega_S]$.\footnote{All logarithms in this paper are base $e$. Furthermore, here and throughout, we take the convention that the operator logarithm is evaluated only on the support of its argument.} This entropy production is always non-negative under the action of thermal operations \cite{fernando2013}. In the absence of work extraction or expenditure, the change of energy in the system is equal to the negated change of energy in the bath 
\begin{equation}
\tr[H_{S} \sigma_{S}]-\tr[H_{S} \rho_{S}]=- Q,
\end{equation}
where $Q$ denotes heat. This is due to energy conservation. 
In the limit of an infinite heat bath, we have that $ Q = \beta \delta S_B$ (the heat and the change of entropy of the bath are proportional). Thus, in that case, we can understand the quantity in \eqref{eq:entprod} as the sum of the change of entropy of the system and bath separately, which is always positive.



How much work could we gain by transforming $\rho_S$ to $\sigma_S$ using such a bath? This question can be answered by asking about the largest value of $W_{\rm gain}(\rho_S \rightarrow \sigma_S) = W$ that can be achieved by a thermodynamical operation belonging to the particular class in question, e.g., thermal operations,
in the transition made by the map in \eqref{eq:TOaction}. The standard second law tells us that this transformation is possible only if
\begin{align}\label{eq:stdSecond}
F(\rho_S \otimes \proj{0}_W) \geq F(\sigma_S \otimes \proj{1}_W), 
\end{align}
where $H_{SW} = H_S + H_W$ and $S(\omega_{SW}) = - \tr[\omega_{SW} \log \omega_{SW}]$.
Using the fact that $\tr[H_W \proj{0}_W] = 0$ and $\tr[H_W \proj{1}_W] = W_{\rm gain}(\rho_S \rightarrow \sigma_S)$, we can use~\eqref{eq:stdSecond} to obtain the following upper bound on the amount of work that we can hope to obtain 
\begin{align}\label{eq:WUpper1}
W_{\rm gain}(\rho_S \rightarrow \sigma_S) \leq F(\rho_S) - F(\sigma_S). 
\end{align} 
 That is, the entropy production upper bounds the deterministic work that can be extracted along a transition $\rho_S \rightarrow \sigma_S$. 
 
 In regimes in which the second law gives necessary and sufficient conditions for particular transitions $\rho_S \rightarrow \sigma_S$ to be possible, it follows that \eqref{eq:stdSecond} can be saturated for any states $\rho_S$ and $\sigma_S$, in which case we have a very tight relation between entropy production and work.
 An example of a regime where \eqref{eq:WUpper1} gives necessary and sufficient conditions occurs if we consider a non-deterministic work paradigm and allow the amount of work to fluctuate arbitrarily, in a transition in which the states are both diagonal in the energy eigenbasis and work is characterised by the mean value of the battery only \cite{skrzypczyk2014work}. Other examples are those in which the systems are extremely large \cite{fernando2013} or if we allow for a slightly inexact catalysis \cite{2ndlaws}. 
Specifically, if an arbitrary catalyst can be used, 
what we mean by this is that
 the ``error per particle'' in the output catalyst is bounded as $\|\eta^{\rm in}_C - \eta_C^{\rm out}\|_1 \leq \varepsilon/\log d_C$, where $d_C$ is the dimension of the catalyst and $\varepsilon>0$ is some tolerance~\cite{2ndlaws}. Similarly, if inexact catalysis takes on the form of allowing small correlations in the output catalyst, only the standard free energy
is relevant~\cite{lostaglio2014e}. A small caveat is that this regime is only achieved in transitions in which both $\rho_S$ and $\sigma_S$ are diagonal in the energy eigenbasis \cite{lostaglio2014thermodynamic}.



It is convenient to note \cite{D87} that the free energy can also be expressed in terms of the quantum relative entropy. 
Specifically, $F(\rho_S) = kT[D(\rho_S\|\tau_S) - \log Z_S]$, where $\tau_S = \exp(-\beta H_S)/Z_S$ is the thermal state of the system at the temperature $T$ of the ambient bath. 
The relative entropy is defined as \cite{umegaki1962conditional}
\begin{equation}
D(\rho\Vert\tau) := \tr[\rho \log \rho] - \tr[\rho \log \tau],
\end{equation}
when $\operatorname{supp}(\rho) \subseteq \operatorname{supp}(\tau)$ and equal to $+\infty$ otherwise.
Since we do not change the Hamiltonian of the system, we can hence express the amount of work in regimes in which the standard free energy is relevant as
\begin{align}
W_{\rm gain}(\rho_S \rightarrow \sigma_S) = kT \Delta,
\label{eq:w_gain_delta}
\end{align}
where we define the difference $\Delta$ of relative entropies, which plays a special role as it is proportional to the entropy production
\begin{equation}
\Delta \equiv D(\rho_S\Vert\tau_S)
 - D(\sigma_S \Vert \tau_S) = \beta 
F(\rho_S)-\beta F(\sigma_S).
\end{equation}

We also define a related quantity, which is the work that needs to be invested in doing the opposite transition, as $W_{\text{inv}}(\sigma_S \rightarrow \rho_S)$. This is how much work is needed to go deterministically from $\sigma_S$ to $\rho_S$.
In the nano-regime, it is possible
 that $W_{\rm gain} \neq W_{\rm inv}$.
 In fact, in general, we have the following relation:
 \begin{align}\label{eq:WUpper2}
 W_{\rm gain}(\rho_S \rightarrow \sigma_S) \leq F(\rho_S) - F(\sigma_S)\ \leq W_{\rm inv}(\sigma_S \rightarrow \rho_S) .
 \end{align} 
This also means that in the regimes in which the free energy gives necessary and sufficient conditions,
$W_{\rm gain}(\rho_S \rightarrow \sigma_S)= W_{\rm inv}(\sigma_S \rightarrow \rho_S)$; i.e., the amount of energy that we need to invest to transform $\rho_S$ to $\sigma_S$ is precisely equal to the amount of work that we can gain by transforming $\sigma_S$ back to $\rho_S$. Thus in this ``standard" free energy regime governed by the Helmholtz free energy $F(\rho_S)$, we see that we do not need to treat the amount of work gained as a separate case, but rather it can be understood fully in terms of the transformation of $\sigma_S$ back to $\rho_S$ in which work needs to be spent. 

It is useful to note that for systems $S$ that are truly small~\cite{HO-limitations}, or when we are interested in the case of \emph{exact} catalysis, 
this is not the case in general. In these situations, the standard second law needs to be augmented with more refined conditions~\cite{2ndlaws} that lead to differences.
With some abuse of terminology, we refer to this as the \emph{nano} regime.
In place of just one free energy, the nano regime requires that a family of free energies $F_\alpha$ satisfies
\begin{align}
F_{\alpha}(\rho_S) \geq F_{\alpha}(\sigma_S),
\end{align}
for all $\alpha \geq 0$. These generalized free energies can be expressed in terms of the $\alpha$-R{\'e}nyi divergences as
\begin{align}
F_{\alpha}(\rho_S) = kT[D_\alpha(\rho_S\|\tau_S) - \log Z_S],
\end{align}
where the general definition of $D_\alpha$~\footnote{For arbitrary states $\rho_S$, we have for $0 \leq \alpha < 1/2$ that $D_\alpha(\rho_S\Vert\tau_S) = \frac{1}{\alpha-1}\log\Tr[\rho_S^{\alpha}\tau_S^{1-\alpha}]$ \cite{P86} and for $\alpha \geq 1/2$, $D_\alpha(\rho_S\Vert\tau_S) = \frac{1}{\alpha-1} \log\left[\tr\left(\tau_S^{(1-\alpha)/(2\alpha)} \rho_S \tau_S^{(1-\alpha)/(2\alpha)} \right)^\alpha\right]$ \cite{MDSFT13,WWY13}.} takes on a simplified form if $\rho_S$ is diagonal in the energy eigenbasis. More precisely,
\begin{align}
D_\alpha(\rho_S\|\tau_S) = \frac{1}{\alpha - 1} \log \sum_{j} \rho_j^\alpha \tau_j^{1-\alpha},
\end{align}
where $\rho_j$ and $\tau_j$ are the eigenvalues of $\rho_S$ and $\tau_S$ respectively.
The standard free energy is a member of this family for $\alpha \rightarrow 1$. 
A short calculation~\cite{2ndlaws} yields that in this regime 
\begin{align}
\label{eq:gainLowerBound}
W_{\rm gain}(\rho_S \rightarrow \sigma_S) &\leq \inf_{\alpha \geq 0} kT\left[D_\alpha(\rho_S\Vert\tau_S) - D_\alpha(\sigma_S\Vert\tau_S)\right], \\
\label{eq:invLowerBound}
W_{\rm inv}(\sigma_S\rightarrow \rho_S) &\geq \sup_{\alpha \geq 0} kT \left[D_\alpha(\rho_S\Vert \tau_S) - D_\alpha(\sigma_S \Vert \tau_S)\right] \nonumber\\
&\geq kT\left[D(\rho_S \Vert \tau_S) - D(\sigma_S \Vert \tau_S)\right],
\end{align}
where (the first) inequalities are again attained if $\rho_S$ is diagonal in the energy eigenbasis.
 
\section{Result}

Our main result is the following relation between entropy production along a change of state and how well a particular change can be undone.
It takes the following form:
\begin{equation}\label{eq:lowbound}
F(\rho_S)-F(\sigma_S)\ge kT \, D(\rho_S \Vert \mathcal{R}_{\sigma \rightarrow \rho}(\sigma_S)),
\end{equation}
where $ \mathcal{R}_{\sigma \rightarrow \rho}$ is a thermal reversal operation using a bath at temperature $T$ that, when $F(\rho_S)-F(\sigma_S)$ is small, takes $\sigma_S$ close to the original state $\rho_S$.
If $\rho_S\rightarrow \sigma_S$ through a map of the form of \eqref{eq:TOaction}, this reversed channel is defined as
\begin{equation}
\mathcal{R}_{\sigma \rightarrow \rho}(\cdot) = \tr_B[U^\dagger ((\cdot)\otimes \tau_B) U].
\end{equation}
That is, the global unitary is reversed after using a new copy of the thermal bath state. This way, in the reversal operation we are ignoring both the correlations with the bath, and its change of state. We now relate the inequality in \eqref{eq:lowbound} to the work relative to the transition $\rho_S \rightarrow \sigma_S$.

{\bf Investing work.}
As outlined above, in the general regime in which not only the standard free energy is relevant, the amount of work $W_{\rm inv}(\sigma_S \rightarrow \rho_S) \geq 0$ we need to invest to transform $\sigma_S$ to $\rho_S$ is larger than the entropy production. Thus \eqref{eq:lowbound} together with \eqref{eq:invLowerBound} guarantee that
\begin{align}\label{eq:investResult}
W_{\rm inv}(\sigma_S\rightarrow \rho_S) \geq kT \, D(\rho_S \Vert \mathcal{R}_{\sigma \rightarrow \rho}(\sigma_S)),
\end{align}
where $\mathcal{R}_{\rho \rightarrow \sigma}$ is again the reversal operation. 
This says, for instance, that if not very much work needs to be spent in restoring $\rho_S$ from $\sigma_S$, then a particular thermal operation not involving any work would also recover $\rho_S$ from $\sigma_S$ well, as measured by the relative entropy distance.

This may not always be the case as for example the erasure of a thermal state $\sigma_S = \tau_S$ to a pure state $\rho_S$ costs a significant amount of work. 
There, the operation $\mathcal{R}_{\rho \rightarrow \sigma}$ will not change the thermal state of the system, effectively not recovering at all. Indeed this inequality also says that if the relative entropy is large, then the amount of work that we need to invest is large too. 
We illustrate this 
 application of our result
 in Section~\ref{sec:example} by means of a simple example of a harmonic oscillator bath.

{\bf Gaining work.} 
Let us focus on particular situations in which $W$ is characterized by the standard free energy (for physical examples of this regime see paragraph after that containing \eqref{eq:WUpper1} in Section \ref{sec:intro}). There, we have that $kT\Delta=W_{\rm gain}(\rho_S \rightarrow \sigma_S)$. In that case \eqref{eq:lowbound} states that the amount of work $W_{\rm gain}(\rho_S \rightarrow \sigma_S) \geq 0$ gained when transforming $\rho_S$ to $\sigma_S$ can be characterized by how well we can recover the state $\rho_S$ from $\sigma_S$ using a thermodynamic operation of the same class which requires \emph{no} work at all. More precisely 
\begin{align}\label{eq:main}
W_{\rm gain}(\rho_S \rightarrow \sigma_S) \geq kT \, D(\rho_S \Vert \mathcal{R}_{\sigma \rightarrow \rho}(\sigma_S)).
\end{align}

In this particular case a link is established between the reversibility of some transition and the amount of work that could be drawn from it. Loosely speaking, if little work can be obtained when transforming $\rho_S$ to $\sigma_S$ with a thermodynamic operation, then there exists a thermodynamic operation of the same class that can recover $\rho_S$ from $\sigma_S$ quite well. Or stated differently, if this thermodynamic operation performs badly at recovering $\rho_S$, 
then the amount of work that can be obtained in the transition $\rho_S \rightarrow \sigma_S$ can be large.  



\section{Proof for thermal operations}

\label{sec:first-result}

We now give details of our main result, which applies to the set of thermal operations (TO) without catalysts.
Section~\ref{sec:ext} contains
details of other, more general sets of operations.

Let us first suppose that we can draw a positive amount of work by transforming $\rho_S$ to $\sigma_S$, so that $\Delta > 0$. 
Note that in regimes dictated by the standard free energy, $\Delta > 0$ implies that there exists a different thermal operation
taking $\rho_S$ to $\sigma_S$ without drawing any work at all \cite{2ndlaws,Lostaglio2015}---in this case the additional energy can be deposited into the bath.
Let $V$ be the energy-conserving unitary that realizes this latter thermal operation, and
let $(\hat{\tau}_B,H_B)$ be the thermal state and Hamiltonian of the bath, such
that $\sigma_S = \tr_B[V(\rho_S \otimes \hat{\tau}_B)V^\dagger]$. Note that $V$ acts on systems $S$ and $B$ and
$[V,H_S + H_B] = 0$.
We have the following theorem:
	
\begin{theorem}\label{th:TO}
	Let $\mathcal{T}$ be a thermal operation given by
	\begin{equation}
	\mathcal{T} (\cdot)_S=\tr_{B}[{V ((\cdot)_S \otimes \hat \tau_B) V^\dagger}],
	\end{equation}
	where $V$ and $\hat{\tau}_B$ are defined above.
	Then it obeys the inequality
	\begin{equation}
	D(\rho_S \|\tau_S)-D(\sigma_S \|\tau_S) \ge D(\rho_S \|\mathcal{R} (\sigma_S)), \label{eq:boundF}
	\end{equation}
	where $\mathcal{R} (\cdot)$ is a recovery channel, which is another thermal operation given by
	\begin{equation}
		\mathcal{R} (\cdot)=\tr_{B}[{V^\dagger ((\cdot)_S \otimes \hat \tau_B ) V}].
	\end{equation}
	
\end{theorem}

\begin{proof}
\smallskip
Our proof is divided into two main steps.

{\bf Step 1: Rewriting the relative entropy difference.} 

Our first step is to rewrite $\Delta=D(\rho_S\|\tau_S) - D(\sigma_S\|\tau_S)$
as an \emph{equality} involving the operation $V$. Observe that%
\begin{align}
D(\rho_{S}\Vert\tau_{S})  & =D(\rho_{S}\otimes\hat{\tau}_{B}\Vert\tau
_{S}\otimes\hat{\tau}_{B})\\
& =D(V(\rho_{S}\otimes\hat{\tau}_{B})V^{\dag}\Vert V(\tau_{S}\otimes\hat{\tau
}_{B})V^{\dag})\\
& =D(V(\rho_{S}\otimes\hat{\tau}_{B})V^{\dag}\Vert\tau_{S}\otimes\hat{\tau
}_{B}),
\end{align}
where we have used the facts that the relative entropy
is invariant with respect to tensoring an ancilla state or 
applying a unitary, and $V$ is an energy-conserving unitary
so that $V(\tau_{S}\otimes\hat{\tau
}_{B})V^{\dag} = \tau_{S}\otimes\hat{\tau
}_{B}$. 


For density operators $\eta_{CD}$ and $\theta_{CD}$ such that
$\operatorname{supp}(\eta_{CD}) \subseteq
\operatorname{supp}(\theta_{CD})$, it is
possible to write 
\begin{equation}\label{eq:diffRewrite}
D(\eta_{CD}\Vert\theta_{CD})-D(\eta_{D}\Vert\theta_{D})
=\Tr(\eta_{CD}[ \log \eta_{CD} - \log \theta_{CD}-\log I_{C}\otimes \eta_{D}+\log
I_{C}\otimes \theta_{D}]). 
\end{equation}
Using these two facts, we can rewrite $\Delta$ as follows:%
\begin{equation} \label{eq:reduction}
D(\rho_{S}\Vert\tau_{S})-D(\sigma_{S}\Vert\tau_{S})
=\Tr\left(
V(\rho_{S}\otimes\hat{\tau}_{B})V^{\dag}
[\log V(\rho_{S}\otimes\hat{\tau}_{B})V^{\dag} - 
\log\tau_{S}\otimes
\hat{\tau}_{B}-\log\sigma_{S}\otimes I_{B}+\log\tau_{S}\otimes I_{B}]\right).
\end{equation}
  We can simplify the operator consisting of the last three terms on the right above as%
\begin{align}
 -\log\tau_{S}\otimes\hat{\tau}_{B}-\log \sigma_{S}\otimes I_{B}+\log
\tau_{S}\otimes I_{B}
& =-\log I_{S}\otimes\hat{\tau}_{B}-\log \sigma_{S}\otimes I_{B}\\
& =-\log \sigma_{S}\otimes\hat{\tau}_{B},
\end{align}
and thus conclude that
\begin{equation}
D(\rho_{S}\Vert\tau_{S})-D(\sigma_{S}\Vert\tau_{S}) = D(V(\rho_{S}\otimes\hat{\tau}_{B})V^{\dag}\Vert\sigma_{S}\otimes\hat{\tau}%
_{B}).
\end{equation}
Hence we have that the right-hand side is equal to
\begin{equation}
D(V(\rho_{S}\otimes\hat{\tau}_{B})V^{\dag}\Vert\sigma_{S}\otimes\hat{\tau}%
_{B})=D(\rho_{S}\otimes\hat{\tau}_{B}\Vert V^{\dag}(\sigma_{S}\otimes\hat
{\tau}_{B})V).
\end{equation}
Putting everything together, we see that%
\begin{equation} \label{eq:recovery-rewrite}
D(\rho_{S}\Vert\tau_{S})-D(\sigma_{S}\Vert\tau_{S})=D(\rho_{S}\otimes\hat
{\tau}_{B}\Vert V^{\dag}(\sigma_{S}\otimes\hat{\tau}_{B})V).
\end{equation}

Thus, the quantity $\Delta$ related to the work gain in \eqref{eq:w_gain_delta} is exactly equal to the ``relative entropy distance'' between the original state $\rho_S \otimes \hat{\tau}_B$ and the state resulting from the following thermal operation:
\begin{equation}
\sigma_S \rightarrow V^{\dag}(\sigma_{S}\otimes\hat{\tau}_{B})V,
\end{equation} 
which consists of adjoining $\sigma_S$ with a thermal state
$\hat{\tau}_{B}$ and performing the inverse of the unitary 
$V$. Note that this statement is non-trivial, since $\sigma_S \otimes 
\hat{\tau}_B \neq V(\sigma_S \otimes \hat{\tau}_B)V^\dagger$. The forward unitary operation $V$
can create correlations between the system and the bath, whereas $V^\dagger$ is applied to a fresh and entirely uncorrelated bath, making it a thermal operation.

\smallskip
{\bf Step 2: A lower bound using the recovery map.}
Due to the fact that the quantum relative entropy can never increase under the action of a partial trace~\cite{lindblad:mono,uhlman:mono}, we can conclude from \eqref{eq:recovery-rewrite} that the following 
inequality holds
\begin{equation} \label{eq:ineqmt}
D(\rho_{S}\Vert\tau_{S})-D(\sigma_{S}\Vert\tau_{S}) \geq D(\rho_{S} \Vert \mathcal{R}_{\sigma \rightarrow \rho}(\sigma_S)),
\end{equation}
where
\begin{align}\label{eq:recoveryOp}
\mathcal{R}_{\sigma \rightarrow \rho}(\sigma_S) = \tr_B [V^{\dag}(\sigma_{S}\otimes\hat{\tau}_{B})V].
\end{align}
This concludes the proof. Note that this operation is a thermal operation, and requires no work. 
\end{proof}


\subsection{Remark: Petz recovery map}
We remark that 
$\mathcal{R}$ is actually a special quantum map, called the Petz recovery map~\cite{pet86,pet88,BK02,HJPW04}.
For a general quantum channel $\mathcal{N}$ and a given density operator $\theta$, this recovery map is defined as
\begin{equation}\label{eq:petz}
\tilde {\mathcal{N}}(\cdot) =
\theta^{1/2} \mathcal{N}^\dag[\mathcal{N}(\theta)^{-1/2} (\cdot)
\mathcal{N}(\theta)^{-1/2} ] \theta^{1/2},
\end{equation}
where $\mathcal{N}^\dag$ is the adjoint of the channel
$\mathcal{N}$ \cite{Wat11}.
 As a consequence, we can conclude that the main conjecture from \cite{winter-li} holds for the special case of thermal operations. We show this in the following lemma.

\begin{lemma}
	The map $\mathcal{R}(\cdot)$ in \eqref{eq:recoveryOp} is the Petz recovery map of the original thermal operation, provided we choose the state $\theta$ in  \eqref{eq:petz} to be the thermal state $\tau_S$.
	\end{lemma}
	\begin{proof}
Consider that to two density operators $\eta$ and
$\theta$ and a quantum channel $\mathcal{N}$, we can associate the relative entropy difference
$D(\eta\|\theta) -
D(\mathcal{N}(\eta)\|\mathcal{N}(\theta))$ and the Petz recovery channel of Eq. \eqref{eq:petz} above.
 For our case, we have that
\begin{equation}
\eta  = \rho_S , \qquad
\theta = \tau_S , \qquad
\mathcal{N}(\cdot)  = \tr_B [V((\cdot)_S\otimes\hat{\tau}_{B})V^{\dag}] ,
\end{equation}
which implies that
$
\mathcal{N}(\theta) = \tau_S.
$
Using the definition of the adjoint, one can show that
\begin{equation}
\mathcal{N}^{\dag} (\cdot)
= \tr_B\left[ \hat{\tau}_B^{1/2} V^{\dag} [ (\cdot)_S \otimes I_B ] V\hat{\tau}_B^{1/2}\right],
\end{equation}
which implies for our case that the Petz recovery channel takes the following form:
\begin{equation}
\tilde{\mathcal{N}}(\cdot) = 
\tau_S^{1/2} \tr_B\left[ \hat{\tau}_B^{1/2} V^{\dag} [
\tau_S^{-1/2}(\cdot)_S\tau_S^{-1/2} \otimes I_B ] V\hat{\tau}_B^{1/2}\right] \tau_S^{1/2} \label{eq:thermal-Petz} .
\end{equation}
We can rewrite this as follows:
\begin{align}
 \tr_B\left[ (\tau_S^{\frac{1}{2}} \otimes \hat{\tau}_B^{\frac{1}{2}}) V^{\dag} [
\tau_S^{-\frac{1}{2}}(\cdot)_S\tau_S^{-\frac{1}{2}} \otimes I_B ] V(\tau_S^{\frac{1}{2}} \otimes \hat{\tau}_B^{\frac{1}{2}})\right] 
& =  \tr_B\left[ (\tau_S \otimes \hat{\tau}_B)^{\frac{1}{2}} V^{\dag} [
\tau_S^{-\frac{1}{2}}(\cdot)_S\tau_S^{-\frac{1}{2}} \otimes I_B ] V(\tau_S \otimes \hat{\tau}_B)^{\frac{1}{2}}\right] \\
& =  \tr_B\left[ V^{\dag}(\tau_S \otimes \hat{\tau}_B)^{\frac{1}{2}}
\tau_S^{-\frac{1}{2}}(\cdot)_S\tau_S^{-\frac{1}{2}} \otimes I_B ] (\tau_S \otimes \hat{\tau}_B)^{\frac{1}{2}}V\right] \\
& = \tr_B [V^{\dag}((\cdot)_S\otimes\hat{\tau}_{B})V],
\end{align}
where we have used that $[V, \tau_S \otimes \hat{\tau}_B]=0$.
\end{proof}

\section{Example for thermal operations}

\label{sec:example}
Let us illustrate the reversal operation $\mathcal{R}_{\sigma \rightarrow \rho}$ by means of a simple example. Let $S$ be a two-level system, with Hamiltonian 
$H_S = E_S \proj{1}$. Let us take $\rho_S=\proj{0}_S$ and $\sigma_S =p_0\proj{0}_S+p_1\proj{1}_S$ with $p_0 \in [1-e^{-\beta E_S},1]$. When $p_0=1/2$, the opposite operation $\sigma_S \rightarrow \rho_S$ corresponds to Landauer erasure.
Recall that the reversal operation associated with the lower bound for the work in \eqref{eq:investResult} is determined by the operation that takes $\rho_S = \proj{0}_S$ to $\sigma_S$ without drawing any work.

For our simple example, consider a bath comprised of a harmonic oscillator $H_B = \sum_{n=0}^\infty E_n \proj{n}_B$ where $E_n = n \hbar \omega$.\footnote{We could have also written $E_n = (2n+1) \frac{\hbar}{2} \omega$, which is the same after re-normalizing. For notational convenience we have subtracted the constant $\frac{\hbar}{2} \omega$.} Note that, for each $n$, the gap between $n$ and $n+1$ is constant: $G = E_{n+1} - E_n = \hbar \omega$. 
To illustrate, let us consider the energy gap of the system to be equal to $E_S = \hbar \omega$---an example in which $E_S$ is a multiple of $\hbar \omega$ is analogous. 

\subsubsection{Transforming $\rho_S$ to $\sigma_S$}
Our first goal is to find the explicit operation that takes $\rho_S$ to $\sigma_S$, which has the effect of mixing the ground state of the system. 
Note that since $U$ conserves energy, $U$ is block diagonal in the energy eigenbasis belonging to different energies. 
More precisely, if the total Hamiltonian $H = H_S + H_B$ is block diagonal $H = \bigoplus_n E_n \Pi_{E_n}$ 
where $\Pi_{E_n}$ is the projector onto the subspace of energy $E_n = n\hbar \omega$ 
spanned by $\ket{0}_S\ket{0}_B$ for $n=0$ and $\{\ket{0}_S\ket{n}_B,\ket{1}_S\ket{n-1}_B\}$ for $n=1,2,3,\ldots$, then $U = \bigoplus_n U_{E_n}$, where $U_{E_n}$ is a unitary acting only on the subspace
of energy $E_n$. That is, $\Pi_{E_n} U_{E_n} \Pi_{E_n} = U_{E_n}$. 

Consider the unitary transformations $U_{E_n}$ defined by the following action:
\begin{align}
&U_{E_0} \ket{0}_S\ket{0}_B= \ket{0}_S\ket{0}_B=:\ket{\Psi_{E_0}},\label{eq:U1}\\
&U_{E_n} \ket{0}_S\ket{n}_B = \sqrt{b} \ket{0}_S\ket{n}_B + \sqrt{1-b} \ket{1}_S \ket{n-1}_B =: \ket{\Psi_{E_n}}\quad \text{for } n=1,2,3,\ldots\ ,\label{eq:U2}\\
&U_{E_n} \ket{1}_S\ket{n-1}_B  = \sqrt{1-b} \ket{0}_S\ket{n}_B - \sqrt{b} \ket{1}_S \ket{n-1}_B =: \ket{\Psi_{E_n}^{\perp}} \quad \text{for } n=1,2,3,\ldots \ ,\label{eq:U3} 
\end{align}
where $0\leq b\leq 1$ is a parameter that will be chosen 
in accordance with the desired target state $\rho_S$ below. 
It is useful to observe that in the subspace $\{\ket{0}_S\ket{n}_B,\ket{1}_S\ket{n-1}_B\}$, the unitary $U_{E_n}$ can be written as 
\begin{align}
U_{E_n} = \left(
\begin{array}{cc}
\sqrt{b} & \sqrt{1-b}\\
\sqrt{1-b} & - \sqrt{b}
\end{array}
\right),
\end{align}
which makes it easy to see that $U = U^\dagger$ is Hermitian.
Note that the states are normalized and 
$\inp{\Psi_{E_n}}{\Psi_{E_n}^{\perp}} = 0$ for $n=1,2,3,\ldots$. The bath thermal state is
\be
\hat\tau_B=\frac{1}{Z_B}\sum_{n=0}^\infty \me^{-n E_S\beta}\proj{n}_B,
\ee
where $Z_B = \sum_{n=0}^\infty e^{-n E_S\beta}=1/(1-\me^{-E_S\beta})$ is the partition function of the bath, and we have used the fact that $E_n = n\hbar \omega = n E_S$.
The unitary thus transforms the overall state as
\begin{align}
U(\proj{0}_S \otimes \hat\tau_B)U^\dagger 
& = \frac{1}{Z_B}\sum_{n=0}^\infty \me^{-n E_S \beta} 
U\left( \ket{0}_S\ket{n}_{B\;\;S}\!\bra{0}_{\;B}\!\bra{n} \right) U^\dag\\
& = \frac{1}{Z_B}\sum_{n=1}^\infty \me^{-n E_S \beta} 
U\left( \ket{0}_S\ket{n}_{B\;\;S}\!\bra{0}_{\;B}\!\bra{n} \right) U^\dag
+\frac{1}{Z_B}
U\left( \ket{0}_S\ket{0}_{B\;\;S}\!\bra{0}_{\;B}\!\bra{0} \right) U^\dag\\
& = \frac{1}{Z_B}\sum_{n=1}^\infty \me^{-nE_S\beta} \proj{\Psi_{E_n}}+\frac{1}{Z_B}\proj{0}_S\otimes\proj{0}_B\\
& =: \sigma_{SB}^0.\label{eq:rho0}
\end{align}
By linearity of the partial trace operation, we have that
\begin{align}
\Tr_B(\sigma_{SB}^0) & = \frac{Z_B-1}{Z_B}\left( b\proj{0}_S+(1-b)\proj{1}_S \right)+\frac{1}{Z_B}\proj{0}_S\\
& = p_0\proj{0}_S+p_1\proj{1}_S, \label{eq:rho0PT}
\end{align}
where
\begin{align}
p_0&=\frac{1}{Z_B}\left((Z_B-1)b+1\right),\label{eq:p 0 1 def}\\
p_1&=1-p_0.
\end{align}
Note that since $0 \leq b \leq 1$, $p_0\in[1/Z_B,1]=[1-\me^{-E_S\beta},1]$.
Solving~\eqref{eq:p 0 1 def} for $b$ gives
\be \label{eq:b in terms of p 0 1}
b=\frac{p_0 Z_B-1}{Z_B-1}.
\ee

\subsubsection{The reversal operation}
Let us now construct the reversal map $\mathcal{R}_{\sigma\rightarrow\rho}$.
This map can be written as
\begin{align}
\mathcal{R}_{\sigma \rightarrow \rho}(\sigma_S) 
= \tr_B\left[U^\dagger (\sigma_S \otimes \hat{\tau}_B) U\right]
= \tr_B\left[U (\sigma_S \otimes \hat{\tau}_B) U^\dagger\right],
\end{align}
where we have used the fact that $U = U^\dagger$.
To evaluate the reversal map for arbitrary $\sigma_S$, let us first note that by a calculation similar to the above
\begin{align}
U(\proj{1}_S \otimes \hat{\tau}_B)U^\dagger &= \frac{1}{Z_B} \sum_{n=0}^{\infty} 
e^{-nE_S \beta} \proj{\Psi_{E_{n+1}}^{\perp}} =: \sigma_{SB}^1.\label{eq:rho1}
\end{align}
Using the linearity of the partial trace, we furthermore observe that
\begin{align}
\tr_B\left[\sigma_{SB}^1\right] = (1-b)\proj{0}_S + b \proj{1}_S.\label{eq:rho1PT}
\end{align}
Using~\eqref{eq:rho0} and \eqref{eq:rho1} together with~\eqref{eq:rho0PT} and \eqref{eq:rho1PT}, we then have
\begin{align}
\mathcal{R}_{\sigma \rightarrow \rho}(\sigma_S) 
&= \tr_B\left[U (\sigma_S \otimes \hat{\tau}_B) U^{\dagger}\right]\\
&= p_0 \tr_B\left[\sigma_{SB}^{0}\right] + p_1 \tr_B\left[\sigma_{SB}^1\right]\\
&= p_0 \left(p_0 \proj{0}_S + p_1 \proj{1}_S\right) + 
p_1 \left(   (1-b) \proj{0}_S + b \proj{1}_S\right)\\
&= P_0^{\mathcal{R}} \proj{0}_S + P_1^{\mathcal{R}} \proj{1}_S,
\end{align}
with
\begin{align}\label{eq:P 0 R def}
P_1^\mathcal{R}&:=1-P_0^\mathcal{R},\\
P_0^\mathcal{R}&:= 
\left(p_0\right)^2 + \left(p_1\right)^2 \frac{Z_B}{Z_B-1}
= 
\left(p_0\right)^2 + \left(1-p_0\right)^2 \me^{E_S \beta},\label{eq:popo}
\end{align}
where we have used the fact that $p_0 + p_1 = 1$ and $Z_B = 1/(1-\me^{-E_S\beta})$.
We can now compute the lower bound for $W_\textup{inv}(\sigma_S\rightarrow \rho_S).$ We find
\begin{align} \label{eq:D cal recovery 1}
W_\textup{inv}(\sigma_S\rightarrow \rho_S)\,& \geq kTD(\rho_S\| \mathcal{R}_{\sigma\rightarrow\rho}(\sigma_S))\\
& =-kT\log P_0^\mathcal{R}.
\end{align}
Plugging in~\eqref{eq:popo} into \eqref{eq:D cal recovery 1} we find
\begin{align}
W_\textup{inv}(\sigma_S\rightarrow \rho_S)
& \geq -kT\log \left[ (p_0)^2+(1-p_0)^2 \me^{E_S \beta} \right],\label{eq:final expression}
\end{align} 
where we recall $p_0\in[1/Z_B,1]=[1-\me^{-E_S\beta},1].$

\subsubsection{Three special cases}
We examine three special cases of~\eqref{eq:final expression}:
\begin{itemize}
	\item[1)] Consider $p_0=1$. In this case we want to form the state $\proj{0}_S$ from the state $\proj{0}_S$. The work invested must clearly be zero in this case. The RHS of
	\eqref{eq:final expression} is also zero, and hence the bound \eqref{eq:investResult} is tight for this case.
	
	\item[2)] Consider $p_0=1/Z_S=1/(1+\me^{-E_S\beta})$. That is, we want to ``recover'' from a thermal state $\sigma_S=\tau_S$, as to get close to a ground state. In this case, the RHS of \eqref{eq:final expression} simplifies to $kTD(\rho_S\| \mathcal{R}_{\sigma\rightarrow\rho}(\sigma_S))=kT\log Z_S$.
	By direct calculation using the 2nd laws (using
	~\eqref{eq:gainLowerBound}--\eqref{eq:invLowerBound}) we find $W_\textup{gain}^\textup{nano}=W_\textup{inv}^\textup{nano}=(\log Z_S)/\beta$ and thus the bound is also tight for this case.
	
	\item[3)] Consider $p_0=1/Z_B$. That is, we want the recovery map to approach to a pure state from the state whose ground state population is the same as the ground state population of the harmonic oscillator bath. In this case,  \eqref{eq:final expression} reduces to $kTD(\rho_S\| \mathcal{R}_{\sigma\rightarrow\rho}(\sigma_S))=-kT\log[1 + \me^{-2 E_S\beta} - \me^{-E_S\beta}]$.
\end{itemize}

\section{Extending to more general operations involving catalysts} \label{sec:ext}

We now prove the following lemma, which highlights the condition that a given map has to obey for the proof of Section \ref{sec:first-result} to still hold.  
\begin{lemma}\label{le:general}
	Let $\mathcal{T}(\cdot)$ be a quantum channel with a full-rank steady state $\tau_S = \mathcal{T}(\tau_S)$ specified as
	\begin{equation}
	\mathcal{T}(\cdot)_S = \tr_E[{U( (\cdot)_S \otimes \rho_E) U^\dagger}],
	\end{equation}
	for some unitary $U$ and an environment state $\rho_E$, such that
	\begin{equation}\label{eq:nocorr}
	U (\tau_S \otimes \rho_E) U^\dagger = \tau_S \otimes \rho_E'.
	\end{equation}
	That is, at the fixed point no correlations with the environment are created. It then holds for an arbitrary initial state $\rho_S$, and $\sigma_S=\mathcal{T}(\rho_S)$ that
	\begin{equation}
	D(\rho_S \|\tau_S)-D(\sigma_S \|\tau_S) \ge D(\rho_S \|\mathcal{R} (\sigma_S)), 
	\end{equation}
	where $\mathcal{R} (\cdot)$ is the Petz recovery map for the channel $\mathcal{T}(\cdot)$, given by
	\begin{equation}
	\mathcal{R}(\cdot) = \tr_E[{U^\dagger ( (\cdot)_S \otimes \rho_E' ) U}].
	\end{equation}
\end{lemma}
\begin{proof}
	Our proof follows similar steps to those of the particular case of thermal operations shown previously.
	We first write 
	\begin{align}\label{eq:relU}
	D(\rho_S \|\tau_S)&=D(\rho_S \otimes \rho_E \|\tau_S \otimes \rho_E) \\
	&=D(U (\rho_S \otimes \rho_E ) U^\dagger \| U (\tau_S \otimes \rho_E )U^\dagger)  \\
	&=D(U (\rho_S \otimes \rho_E) U^\dagger \|\tau_S \otimes \rho_E'),
	\end{align}
	where we have used
	the main assumption of the lemma from \eqref{eq:nocorr} and
	the facts that the relative entropy
	is invariant with respect to tensoring an ancilla state or 
	applying a unitary.
	
	Now we recall the identity of  \eqref{eq:diffRewrite} from the proof of Theorem \ref{th:TO}:	
	\begin{equation}
D(\eta_{CD}\Vert\theta_{CD})-D(\eta_{D}\Vert\theta_{D})
=\Tr(\eta_{CD}[ \log \eta_{CD} - \log \theta_{CD}-\log I_{C}\otimes \eta_{D}+\log
I_{C}\otimes \theta_{D}]),
	\end{equation}
	where $\operatorname{supp}(\eta_{CD}) \subseteq
\operatorname{supp}(\theta_{CD})$.
		We use it together with \eqref{eq:relU} to write
	\begin{equation}\label{eq:eqrel}
	D(\rho_{S}\Vert\tau_{S})-D(\sigma_{S}\Vert\tau_{S})
	=\Tr(U(\rho_{S}\otimes\rho_E)U^{\dag}
	[\log U(\rho_{S}\otimes\rho_E)U^{\dag} - \log\tau_{S}\otimes
	\rho_E'-\log \sigma_{S}\otimes I_{E}+\log\tau_{S}\otimes I_{E}] ).
	\end{equation}	
	The last three terms on the right-hand side above can be simplified significantly
	\begin{align}
	 - \log\tau_{S}\otimes \rho_E'
	 - \log\sigma_{S}\otimes I_{E}
	 +\log
	\tau_{S}\otimes I_{E}
	 & =
	 -\log I_{S}\otimes \rho_E'
	 -\log\sigma_{S}\otimes I_{E}\\
	& = -\log\sigma_{S}\otimes\rho_E',\label{eq:subs}
	\end{align}
	which leads to
	\begin{equation}
	D(\rho_{S}\Vert\tau_{S})-D(\sigma_{S}\Vert\tau_{S}) = D(U(\rho_{S}\otimes\rho_E)U^{\dag}\Vert \sigma_{S}\otimes \rho_E').
	\end{equation}
	 We also have that
		\begin{equation}
	D(U(\rho_{S}\otimes\rho_E)U^{\dag}\Vert \sigma_{S}\otimes \rho_E')=D(\rho_{S}\otimes \rho_E \Vert U^{\dag}(\sigma_{S}\otimes \rho_E')U).
		\end{equation}
	Putting everything together, we see that%
	\begin{align} 
	D(\rho_{S}\Vert\tau_{S})-D(\sigma_{S}\Vert\tau_{S})&=D(\rho_{S}\otimes \rho_E \Vert U^{\dag}(\sigma_{S}\otimes \rho_E')U) \\
	& \ge D(\rho_{S}\Vert \mathcal{R}' (\sigma_{S}) ).
	\end{align}
	What is left is to show that the recovery map $\mathcal{R}$ is indeed the Petz recovery map. Again this follows by the same reasoning as given previously for thermal operations.
	
	The adjoint of the map $\mathcal{T}(\cdot)$ is as follows
	\begin{equation}
	\left(  \cdot\right)  _{S}\rightarrow\tr_{E}\left[ \rho_E^{1/2}U^{\dag}\left(  \left(  \cdot\right)  _{S}\otimes I_E\right)
	U\rho_E^{1/2}\right],
	\end{equation}
	and, by definition the Petz recovery channel is given as
	\begin{equation}
	\left(  \cdot\right)  _{S}\rightarrow \tau_{S}^{1/2}\tr_{E}\left[  \rho_E^{1/2}U^{\dag}\left(
	\tau_{S}^{-1/2}\left(  \cdot\right)  _{S}\tau_{S}^{-1/2}\otimes I_E \right)  U \rho_E^{1/2} \right]  \tau_{S}^{1/2}.
	\end{equation}
	By a series of steps similar to those shown previously, we have that
	\begin{align}
	& \!\!\!\! \tau_{S}^{1/2}\tr_{E}\left[  \rho_E^{1/2}U^{\dag}\left(
	\tau_{S}^{-1/2}\left(  \cdot\right)  _{S}\tau_{S}^{-1/2}\otimes I_E \right)  U \rho_E^{1/2} \right]  \tau_{S}^{1/2}\nonumber\\
	& =\tr_{E}\left[  \left(\tau_S \otimes \rho_E \right)^{1/2} U^{\dag}\left(
	\tau_{S}^{-1/2}\left(  \cdot\right)  _{S}\tau_{S}^{-1/2}\otimes I_E \right)  U \left(\tau_S \otimes \rho_E \right)^{1/2} \right]  \\
	& =\tr_{E}\left[  U^{\dag} \left(\tau_S \otimes \rho_E' \right)^{1/2}\left(
	\tau_{S}^{-1/2}\left(  \cdot\right)  _{S}\tau_{S}^{-1/2}\otimes I_E \right)   \left(\tau_S \otimes \rho_E' \right)^{1/2} U \right]  \\
	& =\tr_{E}\left[  U^{\dag}\left(  \left(  \cdot\right)  _{S} \otimes \rho_E' \right)  U\right]  \\
	& =\mathcal{R} \left(
	\cdot\right)  ,
	\end{align}
	 We note that from \eqref{eq:nocorr}, multiplying by $U$ and $U^\dagger$ and taking the square root on both sides of the equation allows us to conclude $(\tau_S\otimes\rho_E)^{1/2} U^\dagger = U^\dagger(\tau_S\otimes\rho_E')^{1/2}$ and $U (\tau_S\otimes\rho_E)^{1/2} = (\tau_S\otimes\rho_E')^{1/2} U$. These equalities allow us to go from the second to the third line by using the assumption of  \eqref{eq:nocorr}.
\end{proof}

This lemma implies that for any quantum channel that has a dilation satisfying  the condition in \eqref{eq:nocorr}, we arrive at an inequality like that in   \eqref{eq:ineqmt}. In the next lemma, we define a further set of maps for which the  condition in \eqref{eq:nocorr} holds. 
 We say there is an \emph{isentropic catalytic thermal operation} (ICTO) from $\rho_{S}$ to $\sigma_{S}$, if there exists an energy-conserving unitary $V$ acting on the system
$S$, the bath $B$, and a set of $n$ \textit{isentropic catalysts} $\otimes_{i=1}^n \eta_{C_i}=:\eta_C$ on $C= \bigotimes_{i=1}^n C_i $ with initial states $\eta_{C_i}$
, such that
\begin{equation}\label{eq:catcon2}
\tr_{B}\left[  V\left(  \rho_{S} \otimes\hat{\tau}%
_{B}\otimes\eta_{C} \right)  V^{\dag}\right]  =\sigma_{SC},
\end{equation}
where $\tr_{C}[\sigma_{SC}]=\sigma_S$. The unitary $V$  conserves the energy of the bath, the system, and all the catalysts, so that $[V,H_S+H_B+ H_{C}]=0$, where $H_C:=\sum_{i=1}^n H_{C_i}$ are the Hamiltonians of the catalysts. This said, correlations between the different catalysts $\eta_{C_i}$ are allowed in the final state. For every ICTO, we can define an associated channel $\mathcal{H}_S\rightarrow \mathcal{H}_S$: 
\be\label{eq:nICTO channel def}
\mathcal{T}_{ICTO}(\cdot):= \tr_{BC}\left[  \sigma_{SBC}'(\cdot)\right], \quad \sigma_{SBC}'(\cdot):=V\left(  (\cdot) \otimes\hat{\tau}%
_{B}\otimes\eta_{C} \right)  V^{\dag}.
\ee 
The isentropic catalysts are required to satisfy the following:
\begin{itemize} 
	\item [1)] S($\tr_{S \, C \setminus C_l} [{\sigma_{SC}}])=S(\eta_{C_i}) \,\forall \, i$, meaning that the local states of the catalysts return to states of equal entropy to the initial states.
	\item[2)] When the input to the channel in \eqref{eq:nICTO channel def} is the thermal state $\tau_{S}$ of the system, the entropy and mean  energy of the catalysts are non-increasing and non-decreasing respectively: $S(  \eta_{C} )\geq S( \sigma_{C}' )$ and $\tr[ H_C \eta_{C} ]\leq \tr[H_{C} \sigma_{C}' ]$.
\end{itemize}
Condition 1) guarantees that the catalysts are not degraded in the sense of an entropy change, while (as will become evident in the following lemma) condition 2) guarantees that the channel $\mathcal{T}_{ICTO}(\cdot)$ is Gibbs preserving, which is a physically relevant condition for a channel resulting from a thermodynamic process. This new class of operations is between TO and Gibbs preserving maps.

These conditions are different from the ones that apply to catalytic thermal operations as usually defined in the literature \cite{2ndlaws}. For those, a stronger version of condition 1) holds, but condition 2) does not necessarily hold. However, since condition 2) is only required to hold for the von Neumann entropy and mean energy, rather than requiring exact catalysis, it is feasible (given what is known about work embezzlement with inexact catalysts \cite{Ng2014}) that one can always construct a catalyst large enough, such that for every catalytic thermal operation transforming $\rho\rightarrow\sigma$, there exists another catalytic thermal operation also transforming $\rho\rightarrow\sigma$ (possibly with a larger catalyst) such that condition 2) is satisfied. If such a family of catalytic thermal operations exists, it would be very satisfying since via the following lemma it would mean that there is a subset of catalytic thermal operations which allow for all possible transformations as the full set, yet with the additional physically relevant property of belonging to the class of Gibbs preserving maps.

We now show that given the assumptions above, the operations defined as such obey the conditions of Theorem \ref{thm:general-result}.
\begin{lemma}\label{le:commutation}
	For every ICTO channel as defined in \eqref{eq:nICTO channel def} the following hold:
	\begin{itemize}
		\item [1)] The channel is Gibbs preserving: $\mathcal{T}_{ICTO}(\tau_S)=\tau_S$.
		\item[2)] The isentropic catalysts do not become correlated with the bath or the system when the input to the channel is thermal:
			\begin{equation}\label{eq:main result eq Lemma 3}
		V\left(  \tau_{S}\otimes\hat{\tau}%
		_{B} \otimes\eta_{C} \right)  V^{\dag}=\tau_{S}\otimes\hat{\tau}%
		_{B}\otimes\sigma_{C}'.
		\end{equation}
	\end{itemize}
\end{lemma}
\begin{proof}
	Let 
	$\hat \rho_{SB} = \text{Tr}_{C} [V\left(  \tau_{S}\otimes\hat{\tau}%
	_{B}\otimes^n_{i=1}\eta_{C_i}\right)  V^{\dag}]$ denote the local state of the system and the bath after the transformation,   and denote the total Hamiltonian as $H=H_S+H_B+\sum_{i=1}^n H_{C_i}$, the sum of all the local ones. Conservation of energy before and after the operation corresponds to the following:
	\begin{align}
	\text{Tr}[H V\left(  \tau_{S}\otimes\hat{\tau}%
	_{B}\otimes^n_{i=1}\eta_{C_i}\right)  V^{\dag}&]=\text{Tr}[H (\tau_{S}\otimes\hat{\tau}%
	_{B}\otimes^n_{i=1}\eta_{C_i})]=\\
	\text{Tr}[(H_S+H_B) (\tau_S\otimes \tau_B)]+ \text{Tr}[H_{C} \eta_{C}] &\leq\text{Tr}[(H_S+H_B) (\tau_S\otimes \tau_B)]+ \text{Tr}[H_{C} \sigma_{C}'].
	\label{eq:cons-energy}
	\end{align}
	Also, the total average energy is the sum of the local energies
	\begin{equation}
	\text{Tr}[H V\left(  \tau_{S}\otimes\hat{\tau}%
	_{B}\otimes^n_{i=1}\eta_{C_i} \right)  V^{\dag}]=\text{Tr}[(H_S+H_B) \hat \rho_{SB}]+ \text{Tr}[H_{C} \sigma_{C}'],
	\end{equation}
	and hence $\tr[(H_S+H_B) \hat \rho_{SB}]\leq \tr[(H_S+H_B) (\tau_S\otimes \hat\tau_B)]$. Let $T'$ be the temperature of the Gibbs state $\tau_{SB}'$ such that $\tr[(H_S+H_B) \hat \rho_{SB}]=\tr[(H_S+H_B)  \tau_{SB}']$\footnote{Note that such a $T'\geq 0$ always exists since $T'=0$ is the ground state, the Gibbs state mean energy is monotonically increasing with its temperature, and the mean energy of $\hat\rho_{SB}$ is upper bounded by a thermal state of the same Hamiltonian}. For a given fixed energy, the thermal state is the unique state with the highest entropy \cite[Theorem~1.3]{Carlen09}, and hence $S( \tau_{SB}' )\geq S(\hat \rho_{SB})$. Yet, since, $\tr[(H_S+H_B)  \tau_{SB}']\leq \tr[(H_S+H_B) (\tau_S\otimes \hat\tau_B)]$, it follows that the temperature $T$ of state $\tau_S\otimes \hat\tau_B$ satisfies $T\geq T'$, and thus by direct calculation $S(\tau_S\otimes\hat \tau_B )\geq S( \tau_{SB}')$. So we conclude that
	\be\label{eq:entropy catalyst inequality}
	 S(\tau_S\otimes\hat \tau_B )\geq S( \hat\rho_{SB}).
	\ee
	Now we consider the entropy before and after the transformation. Since the joint operation is a unitary, we have from unitary invariance and sub-additivity of quantum entropy  that
	\begin{align}
	S(\tau_S\otimes\hat \tau_B)+ S(\eta_{C})&=S(V\left(  \tau_{S}\otimes\hat{\tau}%
	_{B}\otimes\eta_{C}\right)  V^{\dag}) \\&
	=S(\hat{\rho}_{SBC})\\&
	\leq S(\hat \rho_{SB})+ S(\sigma_{C}')\\&
	\leq S(\hat \rho_{SB})+S(\eta_{C}), \label{eq:entropy}
	\end{align}
	where we have defined $\hat \rho_{SBC}:=V\left(  \tau_{S}\otimes\hat{\tau}_{B}\otimes\eta_{C}\right)  V^{\dag}$. Hence $S(\tau_S \otimes \hat \tau_B )\le S(\hat \rho_{SB})$. Given our conclusion in \eqref{eq:entropy catalyst inequality} regarding conservation of energy, it must then be the case that $S(\tau_S \otimes \hat \tau_B)=S(\hat \rho_{SB})$. At fixed von Neumann entropy, the thermal (Gibbs) state minimises the mean energy, and thus $S(\tau_S \otimes \hat \tau_B)=S(\hat \rho_{SB})$ implies $\tr[(H_S+H_B)  \hat\rho_{SB}]\geq \tr[(H_S+H_B) (\tau_S\otimes \hat\tau_B)]$. Since previously we concluded $\tr[(H_S+H_B)  \hat\rho_{SB}]\leq \tr[(H_S+H_B) (\tau_S\otimes \hat\tau_B)]$, we then have that $\tr[(H_S+H_B)  \hat\rho_{SB}]= \tr[(H_S+H_B)( \tau_S\otimes \hat\tau_B)]$. Thus the last equality together with $S(\tau_S \otimes \hat \tau_B)=S(\hat \rho_{SB})$ implies 
	\be \label{eq:input output equality bath sys}
	\tau_S \otimes \hat \tau_B=\hat \rho_{SB}=\hat \rho_{S}\otimes\hat \rho_{B}.
	\ee 
  Using \eqref{eq:input output equality bath sys}, \eqref{eq:entropy}, and noting that by definition $\hat\rho_C=\sigma_C'$, we conclude
  \be 
  S(\hat\rho_{SBC})\leq S(\hat\rho_S\otimes\hat\rho_B\otimes\sigma_C') = S(\hat\rho_S\otimes\hat\rho_B\otimes\hat\rho_C)\leq  S(\hat\rho_S\otimes\hat\rho_B\otimes\eta_C)=S(\hat\rho_{SBC})
  \ee 
  Hence $ S(\hat\rho_S\otimes\hat\rho_B\otimes\hat\rho_C)=S(\hat\rho_{SBC})$, which is true iff $\hat\rho_{SBC}=\hat\rho_S\otimes\hat\rho_B\otimes\hat\rho_C$. Writing this in terms of $\tau_S$, $\hat\tau_B$, $\sigma_C'$ and $V\left(  \tau_{S}\otimes\hat{\tau}_{B} \otimes^n_{i=1}\eta_{C_i} \right)  V^{\dag}$ gives us \eqref{eq:main result eq Lemma 3}, completing the proof.
\end{proof}

Putting together Lemmas~\ref{le:general} and \ref{le:commutation} and taking the environment state $\rho_E$ from Lemma~\ref{le:general} to be the state of the bath and the set of catalysts (i.e., $\rho_E \equiv \hat \tau_B \otimes \eta_{C}$), we arrive at the following conclusion:
\begin{theorem}\label{thm:general-result}
	Let $\mathcal{T}_{ICTO}(\cdot)$ be an ICTO channel of the form in \eqref{eq:nICTO channel def} given by
	\begin{equation}
	\mathcal{T}_{ICTO}(\cdot)=\tr_{BC}[{U ( (\cdot)_S \otimes \hat \tau_B \otimes \eta_{C} )U^\dagger}].
	\end{equation}
	Then it obeys the inequality
	\begin{equation}\label{eq: pers like inequality applied}
	D(\rho_S \|\tau_S)-D(\sigma_S \|\tau_S) \ge D(\rho_S \|\mathcal{R} (\sigma_S)), 
	\end{equation}
	where $\mathcal{R} (\cdot)$ is the Petz recovery map given by
	\begin{equation}\label{eq:R prime applied}
	\mathcal{R}(\cdot) = \tr_{BC}[{U^\dagger ((\cdot)_S \otimes \hat \tau_B \otimes\sigma_{C}' ) U}].
	\end{equation}
The Petz recovery map preserves the Gibbs state $\mathcal{R}(\tau_S)=\tau_S$.
\end{theorem}
\begin{proof}
	 \eqref{eq: pers like inequality applied} and \eqref{eq:R prime applied} are a direct consequence of Lemmas \ref{le:general} and \ref{le:commutation}. The fact that the Petz recovery map preserves the thermal state follows from Lemma \ref{le:commutation} by inspection.
\end{proof}

%
%

\section{Gibbs preserving maps}\label{sec:Gibbs preserving map}

A general set of maps to which the conditions of Lemma \ref{le:general} do not apply is that of Gibbs preserving maps \cite{faist2015gibbs}, and we hence need a different method to prove an analogous result. To prove a bound in \eqref{eq:boundF}, we use the following general result for quantum maps from \cite{junge2015universal}:

\begin{theorem}
	Let $\mathcal{N}(\cdot)$ be a quantum channel, and let $\eta$ and $\theta$ be  quantum states. We have that
	\begin{equation}\label{eq:junge}
	D(\eta \| \theta) - D(\mathcal{N}(\eta) \| \mathcal{N}(\theta)) \ge - \int_\mathbb{R} \text{d}t\,\, p(t) \log{F(\eta, \ \tilde{\mathcal{N}}_t (\mathcal{N}(\eta)))},
	\end{equation}
	where $F(\rho,\sigma)=(\tr[\sqrt{\sqrt{\sigma}\rho \sqrt{\sigma}}])^2$ is the quantum fidelity, the map $\tilde{\mathcal{N}}_t$ is the following \emph{rotated} recovery map
	\begin{equation}
	\tilde{\mathcal{N}}_t(\cdot)=\theta^{i t/2} \tilde {\mathcal{N}} (\mathcal{N}(\theta)^{-i t/2} (\cdot) \mathcal{N}(\theta)^{i t/2})\theta^{-i t/2},
	\end{equation}
	with $ \tilde {\mathcal{N}}$ defined as in \eqref{eq:petz} 
	and $p(t)=\frac{\pi}{2} (\cosh (\pi t)+1)^{-1}$ is a  probability density function.
\end{theorem}

In the same way as in Theorem \ref{thm:general-result}, it can be seen by inspection that if we take the map
$\mathcal{N}$ to be Gibbs-preserving so that $\tau_S = \mathcal{N}(\tau_S)$ and if we set $\theta=\tau_S$, then the rotated recovery map is Gibbs-preserving as well, namely $\tilde{\mathcal{N}}_t(\tau_S)=\tau_S$. More explicitly, the bound on $\Delta$ is as follows:
\begin{equation}
\Delta=D(\rho_S \| \tau_S) - D(\sigma_S \| \tau_S) \ge - \int_\mathbb{R} \text{d}t\,\, p(t) \log{F(\rho_S, \tilde{\mathcal{N}}_t (\sigma_S))},
\end{equation}
where instead of having the relative entropy, we have the fidelity in the lower bound for the decrease of free energy.

\section{Conclusion}

We have shown how the amount of entropy produced along a thermal process that takes $\rho_S$ to $\sigma_S$  is directly linked to the reversibility of the process. Specifically, we see that if this quantity is small, then there exists a recovery operation that approximately restores the system to its initial state at no work cost at all. For thermal operations, this map is another operation in the same set and can be taken to be $\mathcal{R}_{\sigma \rightarrow \rho}$.

Our main result applies to the decrease of standard free energy, and it is a very interesting open question to extend our result 
to regimes in which we require the full set of second laws~\cite{2ndlaws}. What makes this question challenging is that~\eqref{eq:diffRewrite}
does not carry over to the regime of $D_{\alpha}$ for $\alpha \neq 1$, and indeed recent work~\cite{SBW14} suggests that other quantities 
naturally generalize the \emph{difference} of relative entropies---and this generalization does not always result in the difference of $\alpha$-R{\'e}nyi relative entropies. It hence forms a more fundamental challenge to understand whether the difference of such $\alpha$-relative entropies, or the quantities suggested
in~\cite{SBW14} should be our starting point. However, the quantities in~\cite{SBW14} would require a proof of a new set of second laws. 

We have applied our analysis to the case of investing work, which in the regime where only the standard free energy is relevant can be characterized fully
by how much work can be \emph{gained} by the inverse process. This relation to the inverse process is not true in the nano-regime where all the refined second laws of~\cite{2ndlaws} become relevant. Nevertheless, we have shown that the reversal operation of said inverse process can indeed be used to understand 
the amount of work that needs to be invested, adding another piece to the growing puzzle that is quantum thermodynamics. 

Since this initial work, there have been a series of recent results consisting in giving lower bounds to the decrease of relative entropy in different cases of interest, covering a number of different branches within quantum information theory, such as  
\cite{berta2016entropic,buscemi2016approximate,alhambra2016dynamical,lemm2016information,marvian2016clocks}, which shows the importance of the concept of recoverability and of recovery maps.

\section{acknowledgments}
We thank Jonathan Oppenheim, Christopher Perry, and Renato Renner for insightful discussions. 
MMW is grateful to SW and her group for hospitality during
a research visit to QuTech, and acknowledges support from startup
funds from the Department of Physics and Astronomy at LSU, the NSF\ under
Award No.~CCF-1350397, and the DARPA Quiness Program through US Army Research
Office award W31P4Q-12-1-0019.
MPW and SW acknowledge support from MOE Tier 3A Grant MOE2012-T3-1-009 and STW, Netherlands.

\end{document}